%% file: conference.tex
\documentclass[10pt, conference, letterpaper]{IEEEtran}
\IEEEoverridecommandlockouts

\usepackage{cite}
\usepackage{amsmath,amssymb,amsfonts}
\usepackage{amsthm}
\usepackage{graphicx}
\usepackage{textcomp}
\usepackage[dvipsnames,table]{xcolor}
\usepackage{colortbl}
\usepackage{subcaption}
\usepackage{algorithm}
\usepackage{algorithmic}
\usepackage{multirow}
\usepackage{array}
\usepackage{makecell}
\usepackage{caption}
\usepackage{booktabs}
\usepackage{ragged2e}
\usepackage{pifont}
\usepackage{xspace}
\usepackage[draft]{hyperref}
\usepackage{cleveref}

\usepackage{tcolorbox}
\tcbset{before skip=3pt, after skip=4pt, top=2pt, bottom=2pt, boxsep=3pt, left=5pt, right=5pt}
\newtheorem{observation}{Observation}
\newtheorem{proposition}{Proposition}

\colorlet{phaseI}{rgb:red!2,65;green!30,60;blue!20,125}
\colorlet{phaseII}{rgb:red!2,65;green!30,90;blue!20,125}
\colorlet{phaseIII}{rgb:red!60,100;green!20,90;blue!30,125}

\newcommand{\Comment}[1]{\hfill$\triangleright$ \textit{#1}}

\newcommand{\first}[1]{\textbf{#1}}

\definecolor{deepred}{rgb}{0.6,0,0}
\definecolor{crystalblue}{rgb}{0.20,0.40,0.80}
\definecolor{darkgreen}{RGB}{0, 181, 18}
\definecolor{darkred}{RGB}{252, 90, 90}
\newcommand{\gain}[1]{$\uparrow$ #1}
\newcommand{\mygreen}[1]{\cellcolor{darkgreen!#1}}
\definecolor{DeltaBg}{HTML}{D4F2D7}
\definecolor{SearchBg}{HTML}{C2E6F5}
\definecolor{AgenticBg}{HTML}{F5C2CC}
\definecolor{MathBg}{HTML}{E6D4F2}
\definecolor{ScienceBg}{HTML}{FBE0BC}

\newcommand{\sys}{FedSceneX\xspace}

\newcommand{\best}[1]{\textcolor{red}{\textbf{#1}}}

\newcommand{\std}[1]{$_{\pm\text{\scriptsize #1}}$}

\newcommand{\venue}[1]{{\scriptsize\color{gray}[#1]}}

\usepackage{stfloats}
\usepackage{balance}
\graphicspath{{figures/}}

\begin{document}

\bstctlcite{IEEEexample:BSTcontrol}

\title{FedSceneX: Time-to-Target Orchestration for Same-Scene Multimodal Federated Edge Learning}

\author{
\IEEEauthorblockN{Dhe yeong Tchalla\IEEEauthorrefmark{1},
Beining Wu\IEEEauthorrefmark{1},
Jun Huang\IEEEauthorrefmark{1},
Shuyang Gu\IEEEauthorrefmark{2},
and Qiang Duan\IEEEauthorrefmark{3}}
\IEEEauthorblockA{\IEEEauthorrefmark{1}Department of Electrical Engineering and Computer Science, South Dakota State University, Brookings, SD, USA\\
\IEEEauthorrefmark{2}Subhani Department of Computer Information Systems, Texas A\&M University-Central Texas, Killeen, TX, USA\\
\IEEEauthorrefmark{3}Department of Information Sciences and Technology, The Pennsylvania State University, Abington, PA, USA\\
Email: \{Dheyeong.Tchalla, Wu.Beining\}@jacks.sdstate.edu, Jun.Huang@sdstate.edu, s.gu@tamuct.edu, qxd2@psu.edu}
}

\maketitle

\begin{abstract}
Federated learning at the sensing edge is typically evaluated by communication rounds, yet a round does not represent a fixed amount of work. Even on identical hardware, the methods we compare require 3.3 to 9.8 hours per round, which makes round-based comparisons misleading. The problem is more obvious for same-scene multimodal clients, since camera, video, LiDAR, and radar workloads differ substantially in training and communication cost, while existing methods treat the modality composition of each round as fixed. To address it, we introduce \sys{}, an orchestrator that jointly determines round composition to maximize learning value per active hour. The optimization method, \emph{Value-per-Hour Pricing} (VHP), converts the fractional objective through a parametric transformation and dualizes the uplink constraint, yielding a closed-form client price whose weights capture resource shadow costs. Based on these prices, \sys{} selects clients subject to a modality coverage constraint, allocates precision through reverse water filling, and assigns updates to edge servers. On the full nuScenes benchmark with fifteen clients and twelve baselines, \sys{} reduces the active time per round to 3.31 hours, compared with 4.85 to 9.78 hours for the baselines. Across all random seeds, it achieves the highest accuracy within a twenty-hour budget while preserving all four modalities. Its advantage persists from ten to forty-five hours, after which conventional methods overtake it.
\end{abstract}

\begin{IEEEkeywords}
Federated learning, multimodal learning, edge intelligence, network orchestration, client selection, resource allocation, autonomous driving.
\end{IEEEkeywords}

\input{sections/introduction}
\input{sections/related_work}
\input{sections/motivation}
\input{sections/system_design}
\input{sections/evaluation}
\input{sections/conclusion}

\bibliographystyle{IEEEtran}
\bibliography{ref}

\end{document}

%% file: sections/introduction.tex
\section{Introduction}

\IEEEPARstart{F}{ederated} learning trains a shared model without moving raw data off the devices that hold it~\cite{McMahan2017AISTATS,Kairouz2021FTML,Duan2023COMST,Wu2025WASA,Wu2026ARXIVForget,Wu2026ARXIVCrystalMem}. In autonomous driving and urban sensing, those devices are the sensors of one platform. Cameras, video streams, LiDAR, and radar observe the same scene at the same instant and share a scene-level label, but they agree on little else: their encoders differ in depth and cost, their updates differ in size, and the semantic classes they resolve differ by sensor~\cite{Feng2023ARXIV,Sun2024ECCV,Fang2025JSAC,Fang2025ARXIV,Tchalla2026ACR,Wu2026ARXIVPRISM}.

Two lines of work make such a federation practical. Heterogeneity-aware and personalized optimizers correct client drift or split the model, through dynamic regularization, alternating optimization, adaptive server updates, or a personalized head~\cite{Acar2021ICLR,Gong2022ICDE,Reddi2021ICLR,Oh2022ICLR,Li2021CVPR,Ding2025IPCCC,Ding2026TAAS,Ding2026ICNC,Wu2023MPE}, while communication-efficient methods shrink or ration the payload by quantizing updates, restructuring them, or admitting only the clients that meet a deadline~\cite{Reisizadeh2020AISTATS,Konecny2016ARXIV,Nishio2019ICC,Wu2025ToN,Huang2025TMC,Wu2026COMST,Wu2023ACCESS,Dong2026TCCN,Pudasaini2026HPSR,Wu2025RACS,Xing2026ACR,Pan2023SCIS,Fang2026ARXIVLLMSearch}. Both are measured per communication round, and that unit is the problem. A round is not a fixed amount of work: its duration is set by which clients are selected and how much local computation each is asked to perform, so on identical hardware and under the same local budget the methods we compare spend between 3.3 and 9.8 hours in one round. Two consequences follow. First, \emph{evaluation} misreports progress, because a convergence curve drawn against the round index compares quantities of unequal price. Second, \emph{control} is aimed at the wrong resource, because a round's time is spent almost entirely on local computation, so a policy that only compresses what is sent cannot change how long the round takes.

Same-scene multimodal learning turns this from an accounting problem into a design one, since the composition of a round now decides its price. A round of cameras is cheap and semantically narrow, a round with video and LiDAR is expensive and informative, and a round that omits radar gives up a quarter of the sensing problem. Multimodal federated learning studies how such clients should be aligned, fused, or reconstructed when a modality is missing~\cite{Feng2023ARXIV,Yuan2023ARXIV,Sun2024ECCV,Wu2026ARXIV,Ding2026ARXIVEASE}, but it takes the modality mixture of a round as given rather than choosing it. Client selection comes closest, scoring clients by a utility that mixes statistical value with systems cost~\cite{Lai2021OSDI}, yet the score decides participation alone, so the precision of an admitted update and the server that receives it stay outside the policy. Thus, an orchestrator that has to reach a usable model within a fixed number of hours has no single quantity to optimize, and instead tunes a selection threshold, a compression schedule, and a placement heuristic that cannot trade one budget against another.

To address this, we introduce \sys{}, a modality-aware orchestrator that treats the composition of a round as the decision variable. \sys{} leaves the local objective, the aggregation rule, and the model architecture unchanged, and modality-specific encoders never leave the device. Each round profiles clients with a few privacy-safe scalars, prices each client with one number, and issues a \emph{contract} that fixes whether the client uploads, at what precision it transmits, and where its update is aggregated. Unlike prior orchestration, the objective \sys{} maximizes is learning value per active hour rather than per round, which makes the price of a round part of the decision instead of an outcome of it.

Optimizing a ratio of this kind is harder than optimizing a sum, and the round's active time is a makespan that does not separate across clients. We propose \emph{Value-per-Hour Pricing} (VHP), which removes both obstacles. Instead of scoring clients by a weighted sum chosen by hand, VHP applies a parametric transform to the fractional objective and dualizes the uplink budget, so the coupled program collapses into one closed-form scalar per client whose two negative weights are the shadow price of a byte and the worth of an hour. The same scalar then answers all three questions: a positivity threshold admits clients under a modality-coverage floor, a reverse water-filling rule quantizes every admitted update to a common distortion floor set by the uplink price, and a least-loaded rule places it on an edge server with a bounded loss in makespan. One number replaces the separate thresholds a modality-aware deployment would otherwise tune by hand.

\begin{itemize}
\item We identify the price of a round as the variable that federated evaluation omits, and we formulate round composition as the maximization of learning value per active hour under an uplink budget and a modality-coverage floor.

\item We propose VHP, which turns that fractional program into one closed-form price per client, and we read that price three ways: as a participation threshold, as a reverse water-filling rule for precision, and as a placement rule whose makespan is within a bounded factor of the optimum.

\item On the complete nuScenes benchmark with fifteen heterogeneous clients and twelve baselines, we report the boundary of the claim together with the claim. A \sys{} round costs $3.31\pm0.13$ active hours against $4.85$ to $9.78$, every seed leads the strongest baseline within a twenty-hour budget, and the lead reverses beyond forty-five hours. An ablation separates the contract into the axes it acts on, and the scalars that make pricing possible identify a client's modality with accuracy 1.00.
\end{itemize}

Contrary to the intuition that a communication-constrained system is made faster by sending less, the network is not where a round spends its time: the entire network cost of a \sys{} round is 1.99 seconds, or 0.016\% of it, so compression acts on a fraction too small to recover. To our knowledge, \sys{} is the first system to price participation, precision, and placement for same-scene multimodal clients with one utility, and to report the operating region in which doing so is the right choice.

%% file: sections/related_work.tex
\section{Related Work}
\label{sec:related}

\subsection{Heterogeneity-Aware and Personalized Federated Learning}
To train on data that never leaves the device, federated learning averages local updates instead of gradients over pooled data~\cite{McMahan2017AISTATS,Kairouz2021FTML,Duan2023COMST}. Li \textit{et al.}~\cite{Li2020MLSYS} propose \emph{FedProx}, which adds a proximal term to limit client drift, and later work replaces the server rule or the local objective with dynamic regularization, alternating optimization, elastic aggregation, and adaptive server updates~\cite{Acar2021ICLR,Gong2022ICDE,Chen2023CVPR,Reddi2021ICLR}. A second group splits the model instead, keeping a personalized head, local normalization statistics, or a shared representation on the device~\cite{Oh2022ICLR,Li2021ICLR,Collins2021ICML}, and a third regularizes the representation against skewed or shifted client distributions~\cite{Li2021CVPR,Guo2023ICML,Nguyen2022NeurIPS,Wu2026TNSE,Wu2026ARXIVLifecycle,Wu2026ICDCS}. These methods change what a round optimizes, but they leave the price of the round outside the policy. \sys{} is complementary to all of them: it changes what a round costs, and it leaves the local objective and the aggregation rule untouched.

\subsection{Multimodal Federated Learning}
Because the clients on one sensing platform carry different sensors, their encoders and their updates are no longer comparable. Sun \textit{et al.}~\cite{Sun2024ECCV} propose \emph{FedCoLa}, which aligns modality-specific transformers across clients, while benchmarks and fusion methods study missing modalities and selective modality communication~\cite{Feng2023ARXIV,Yuan2023ARXIV,Wu2026ARXIV,Wu2026ARXIVPRISM,Ding2026ARXIVEASE}. These methods make heterogeneous modalities learn together, but they take the modality mixture of a round as given. They are compatible with \sys{}, which schedules that mixture and can carry any of their fusion rules.

\subsection{Budgeted Round Composition}
Closest to ours, a few systems compose a round by what it will cost rather than by what it will learn. Lai \textit{et al.}~\cite{Lai2021OSDI} propose \emph{Oort}, which is similar to ours in scoring clients by a utility that mixes statistical value with systems cost, but the score selects participants only, so precision and placement stay outside the policy and modality coverage is not a constraint. Deadline-based admission~\cite{Nishio2019ICC}, periodic averaging with quantized updates~\cite{Reisizadeh2020AISTATS}, and structured-update compression~\cite{Konecny2016ARXIV} each cap one resource in isolation. \sys{} differs in its goal: one utility prices participation, precision, and placement together, against the hours a deployment spends rather than the rounds it completes.

%% file: sections/motivation.tex
\section{Background and Motivation}
\label{sec:motivation}

\subsection{Multimodal Federated Learning at the Edge}

\noindent\textbf{Same-scene multimodal clients.}
In autonomous driving and urban sensing, several sensors observe one scene at the same instant. On nuScenes~\cite{Caesar2020CVPR}, six cameras, two video streams, two LiDAR partitions, and five radars form fifteen federated clients that share a scene-level label but not a representation. Their encoders differ in depth and cost, their updates differ in size, and their evidence differs by semantic class. Federated learning keeps the raw streams on the device~\cite{McMahan2017AISTATS,Kairouz2021FTML,Wu2026MNET,Ding2026ARXIVTwinLoop,Ding2026ICDCS}, and multimodal federated learning aligns or fuses the resulting representations~\cite{Feng2023ARXIV,Yuan2023ARXIV,Sun2024ECCV}. Neither decides how much of a training round each modality is allowed to consume.

\noindent\textbf{What the server may observe.}
Data locality restricts the controller to scalars. A client can report how far its loss fell, how large its update is, and what uplink it measured, but not what it saw. Any orchestration policy has to be built on this narrow channel.

\subsection{Observation: A Round Is Not a Unit of Progress}
\label{obs:round}

\begin{tcolorbox}
\begin{observation}
A round costs what its clients cost, so ranking methods by the rounds they need reorders them relative to the time those rounds take.
\end{observation}
\end{tcolorbox}

\begin{table}[t]
\centering
\caption{Thirty active hours buy three rounds for the baselines and nine for \sys{}. \colorbox{DeltaBg}{$\Delta$}: Macro-F1 points at that budget.}
\label{tab:roundcost}
\footnotesize
\setlength{\tabcolsep}{3.5pt}
\renewcommand{\arraystretch}{1.15}
\begin{tabular}{l|c|cc|c}
\toprule
& \cellcolor{SearchBg}\textbf{Cost}
& \multicolumn{2}{c|}{\cellcolor{AgenticBg}\textbf{Within 30 h}}
& \cellcolor{DeltaBg}\textbf{$\Delta$} \\
\cmidrule(lr){2-2}\cmidrule(lr){3-4}
\textbf{Method} & \textbf{h/rnd} & \textbf{Rounds} & \textbf{Ma-F1} & \\
\midrule
FedBABU~\cite{Oh2022ICLR} & 9.75 & 3 & 0.6447 & \mygreen{22}\gain{1.12} \\
\rowcolor{crystalblue!5}
FedADMM~\cite{Gong2022ICDE} & 9.26 & 3 & 0.6486 & \mygreen{15}\gain{0.73} \\
MOON~\cite{Li2021CVPR} & 6.20 & 3 & 0.6481 & \mygreen{16}\gain{0.78} \\
\midrule
\rowcolor{crystalblue!10}
\textbf{\sys{}} & \best{3.31}\std{0.13} & \best{8.6}\std{0.5} & \best{0.6559}\std{0.0057} & -- \\
\bottomrule
\end{tabular}
\end{table}

\begin{figure}[t]
\centering
\begin{subfigure}[t]{0.49\columnwidth}
\centering
\includegraphics[width=\linewidth]{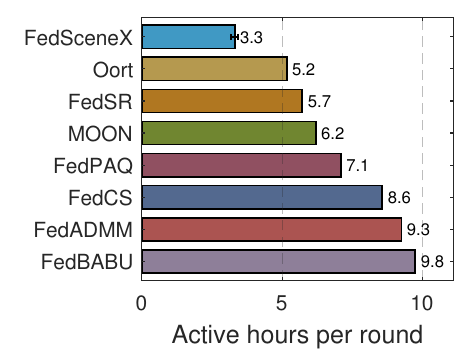}
\caption{Cost of a round}
\label{fig:cost_a}
\end{subfigure}\hfill
\begin{subfigure}[t]{0.49\columnwidth}
\centering
\includegraphics[width=\linewidth]{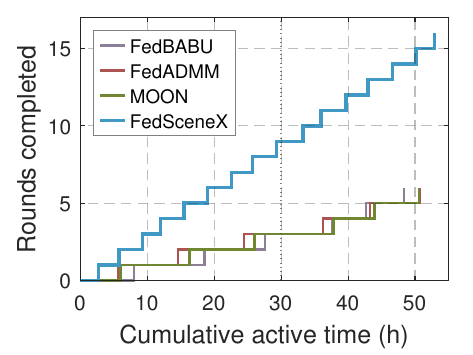}
\caption{What the cost buys}
\label{fig:cost_b}
\end{subfigure}
\caption{(a) Active hours per round, \sys{} over five seeds. (b) Rounds completed against active time; the dotted line marks thirty hours.}
\label{fig:cost}
\end{figure}

\Cref{tab:roundcost} makes the disagreement concrete. Given thirty active hours, FedBABU finishes three rounds and reaches Macro-F1 $0.6447$, while \sys{} finishes nine and reaches $0.6559$. Counted in rounds, \sys{} looks like the slower learner, since it needs three times as many of them to arrive at a similar place. Counted in the hours a deployment actually spends, it arrives first, because a \sys{} round costs $3.31$ hours where theirs cost $6.20$ to $9.75$.

Where a round's time goes explains why. Within a \sys{} round, local training accounts for $95.1\%$ of the wall-clock time and client-side evaluation for a further $4.9\%$, while the entire network cost of the round, base latency, queueing, and transmission together, is 1.99 seconds, or $0.016\%$ of it. Compression acts on that fraction alone. This does not make uploaded volume irrelevant, since uplink capacity is shared and metered, but a method cannot buy wall-clock time by compressing what already takes two seconds. What is left is the composition of the round, namely which modality encoders run and for how long.

\noindent\textbf{Opportunities of composing rounds.}
If a round costs what its clients cost, an orchestrator should choose the composition of each round rather than only the number of participants. Modality then becomes a scheduling variable instead of a property of the dataset.

\subsection{Observation: One Utility Prices Three Decisions}
\label{obs:oneutility}

\begin{tcolorbox}
\begin{observation}
Who uploads, at what precision, and where the update is aggregated are three forms of one question: whether a client's update is worth the resources the round will spend on it.
\end{observation}
\end{tcolorbox}

\begin{table}[t]
\centering
\caption{Each row replaces one family of contract decisions with a uniform policy. \colorbox{DeltaBg}{$\Delta$}: advantage of the contract, in each row's units.}
\label{tab:oneutility}
\footnotesize
\setlength{\tabcolsep}{4.5pt}
\renewcommand{\arraystretch}{1.15}
\begin{tabular}{ll|cc|c}
\toprule
& & \multicolumn{2}{c|}{\cellcolor{ScienceBg}\textbf{Policy}}
& \cellcolor{DeltaBg}\textbf{$\Delta$} \\
\cmidrule(lr){3-4}
\textbf{Decision} & \textbf{Metric} & \textbf{Uniform} & \textbf{Contract} & \\
\midrule
Participate & Active h/round & 9.21 & \first{3.52} & \mygreen{23}\gain{5.69} \\
\rowcolor{crystalblue!5}
Precision & Upload MB/round & 7.27 & \first{2.88} & \mygreen{18}\gain{4.39} \\
Place & Active h/round & 6.60 & \first{3.52} & \mygreen{12}\gain{3.08} \\
\bottomrule
\end{tabular}
\end{table}

These decisions are usually taken by different mechanisms: a selection rule, a compression schedule, and a placement heuristic, each with a threshold of its own. Every threshold has to be retuned when the modality mixture changes, and none of them can trade one budget against another. One utility avoids this. It scores what a client's update is expected to be worth against what sending it will cost, and the same number then answers all three questions. \Cref{tab:oneutility} shows what each family of decisions is worth on the axis it controls. Replacing utility-driven participation with uniform participation raises the round from 3.52 to 9.21 hours, uniform precision raises the uploaded state from 2.88 to 7.27~MB per round, and removing utility-driven placement raises the round to 6.60 hours.

The design has a cost that is easy to overlook. The scalars the controller needs are the scalars that describe a client's workload, and workloads differ systematically by modality. A logistic-regression attacker with access to nothing but this server-visible telemetry recovers a client's modality with $1.00$ accuracy against a majority baseline of $0.36$. Orchestration of this kind buys time with information, and the information is not free.

\noindent\textbf{Opportunities of one-utility orchestration.}
Pricing three decisions with one number removes the separate thresholds a modality-aware deployment would otherwise tune by hand, and it makes the exposure of the control channel a measurable quantity rather than an assumption.

%% file: sections/system_design.tex
\section{System Design}
\label{sec:design}

\subsection{Overview}

\begin{figure*}[t]
\centering
\includegraphics[draft=false,width=0.95\textwidth]{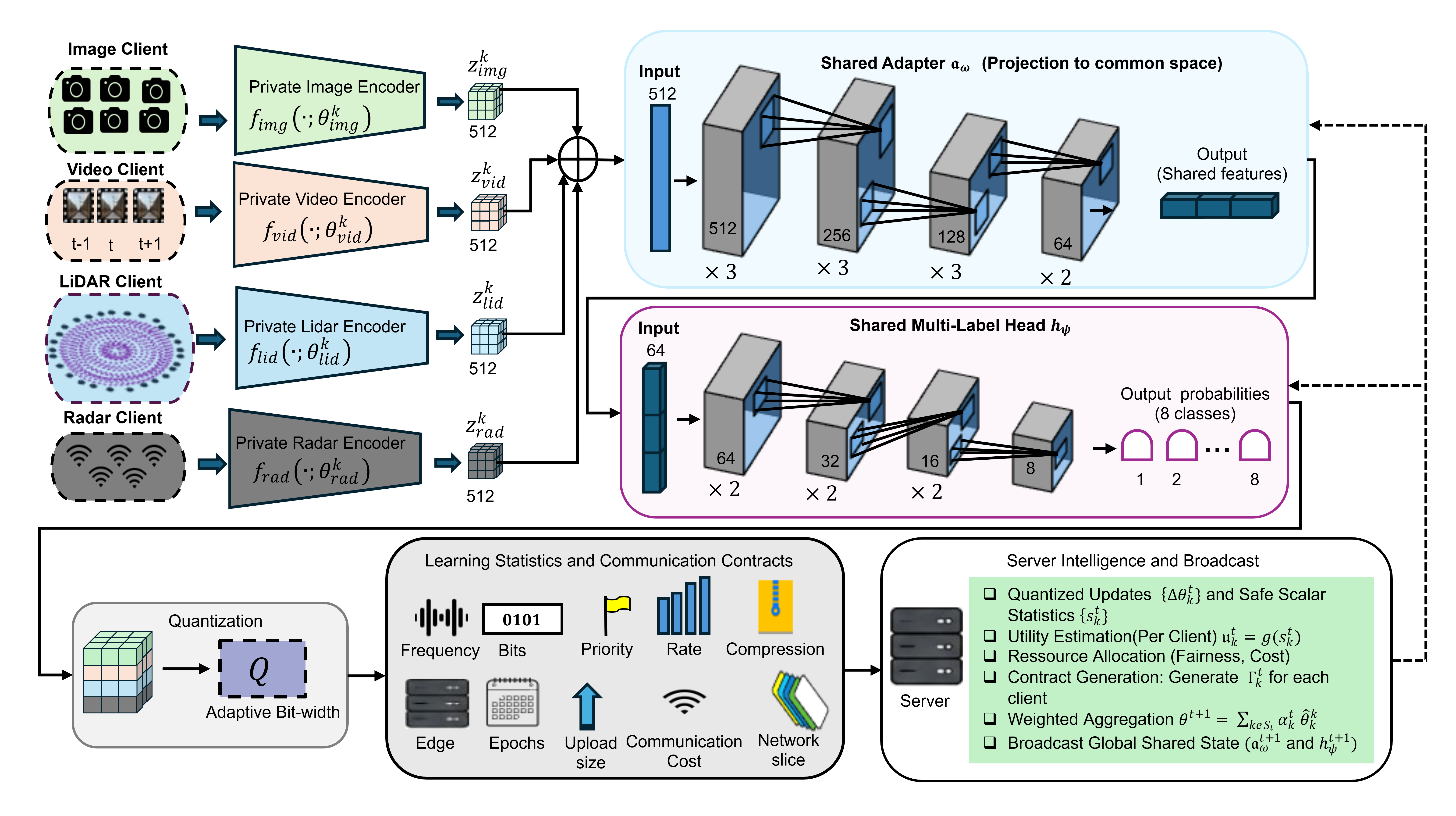}
\caption{Overview of \sys{}. Private encoders stay on the device; one utility over privacy-safe scalars fixes participation, precision, and placement.}
\label{fig:overview}
\end{figure*}

\sys{} is an orchestration layer over federated averaging. It decides only which clients train and transmit in a round and at what cost. Every round runs three phases (\Cref{fig:overview}). In the \emph{profile} phase, each client performs a short local warm-up and reports a few privacy-safe scalars: how far its local loss fell, how large its shared-state update is, what uplink it measured, and how often it has participated. In the \emph{price} phase, \emph{Value-per-Hour Pricing} (VHP) maps those scalars to one learning-aware utility. In the \emph{assign} phase, the server reads that utility three times and issues a \emph{contract} that fixes whether the client uploads, at what precision it transmits, and where its update is aggregated.

The privacy boundary follows the model split. Client $k$ observes modality $m_k$, holds a local dataset $\mathcal{D}_k$, and keeps a private encoder $f_{m_k}(\cdot;\phi_k)$ that never leaves the device. Only the adapter $a(\cdot;\omega)$ and the multi-label head $h(\cdot;\psi)$ are synchronized, so the federated model is $\boldsymbol{\theta}=\{\omega,\psi\}$ and the client prediction is
\begin{equation}
\hat{\mathbf y}_k=h_{\psi}\!\left(a_{\omega}\!\left(f_{m_k}(\mathbf x_k;\phi_k)\right)\right).
\label{eq:client_prediction}
\end{equation}
Raw observations, labels, embeddings, logits, and per-sample outputs stay local. The server receives the quantized update of $\boldsymbol{\theta}$ together with the scalars above, so a client exposes only its $d_s$ shared parameters and not the $d_p$ of its private encoder.

\subsection{The Round as a Ratio}
\label{sec:ratio}

A deployment is judged by the accuracy it reaches per hour and not by the accuracy it reaches per round, so what a round should maximize is a ratio rather than a sum. Let $V(\Gamma^t)$ be the expected worth of the updates a contract set $\Gamma^t$ buys and let $T(\Gamma^t)$ be the active time it consumes, which is the makespan of the selected clients over the edge servers plus their synchronization. The orchestration problem is
\begin{equation}
\begin{aligned}
\max_{\Gamma^t}\quad & \frac{V(\Gamma^t)}{T(\Gamma^t)} \\[2pt]
\text{s.t.}\quad & \sum_{k\in\mathcal{K}} q_k^t b_k^t \le B, \quad
\sum_{k:\,m_k=m} q_k^t \ge 1 \;\; \forall m\in\mathcal{M},
\end{aligned}
\label{eq:ratio}
\end{equation}
together with one edge server per selected client and the discrete range of every contract variable. The second constraint is a coverage floor, and it is what keeps a modality in the round during the stretches when its clients price poorly.

\Cref{eq:ratio} is a fractional program over binary and discrete variables, and neither the ratio nor the makespan inside $T$ separates across clients. Both obstacles have a standard remedy. For the ratio, the parametric transform of Dinkelbach~\cite{Dinkelbach1967MS} replaces $\max V/T$ by the root of
\begin{equation}
\Phi(\kappa)=\max_{\Gamma^t}\big\{V(\Gamma^t)-\kappa\,T(\Gamma^t)\big\},
\label{eq:dinkelbach}
\end{equation}
which is convex and strictly decreasing, so a ratio problem becomes a sequence of problems linear in worth and in time, and $\kappa$ carries the units of value per hour. For the makespan, assigning each selected client to the least-loaded server is list scheduling, whose makespan is within $2-1/|\mathcal{E}|$ of the optimum~\cite{Graham1966BSTJ}, or a factor $1.5$ for the two servers used here, so $T$ may be replaced by the total time the round buys at a bounded loss. What is left couples clients only through the uplink budget, and one multiplier removes it.

\subsection{Value-per-Hour Pricing}
\label{sec:utility}

VHP is the per-client score that the two reductions leave behind.

\begin{proposition}[Separable pricing]
\label{prop:price}
Let $\lambda\ge0$ multiply the uplink budget of \Cref{eq:ratio} and let $\kappa$ be the parameter of \Cref{eq:dinkelbach}. The parametric Lagrangian then decomposes across clients as $\sum_{k}q_k^t U_k^t+\lambda B$ with
\begin{equation}
U_k^t=\alpha P_k^t+\gamma G_k^t+\mu F_k^t-\delta X_k^t-\eta L_k^t ,
\label{eq:utility}
\end{equation}
where $\delta=\lambda$ and $\eta=\kappa$, and the participation that maximizes it is the threshold rule $q_k^{t\star}=\mathbf{1}[U_k^t>0]$ corrected by the coverage floor.
\end{proposition}

\begin{IEEEproof}
Collect what one selected client contributes to $V-\kappa T$: its expected worth $\alpha P_k^t+\gamma G_k^t+\mu F_k^t$ and the time $\kappa L_k^t$ it adds to the round. Adjoining the budget with $\lambda$ adds $-\lambda X_k^t$ to the same bracket and the constant $\lambda B$. Every remaining constraint is a per-client range, so the maximization distributes over $k$ and a client raises the sum exactly when its bracket is positive.
\end{IEEEproof}

\Cref{eq:utility} is therefore not a weighted sum chosen by hand. Its two negative weights are prices: $\delta$ is the shadow price of a byte of uplink, and $\eta$ is the worth of an hour on the round's critical path. This is why \Cref{eq:ratio} carries no second penalty on communication or delay, and why one number can arbitrate between a fast client and an informative one. Here $P_k^t$ is normalized local progress, $G_k^t=\|\Delta\boldsymbol{\theta}_k^t\|_2$ is update magnitude, $F_k^t$ credits clients and modalities that have participated less, $X_k^t$ is the uplink cost of the update after compression, and $L_k^t$ is the time the client adds to the round, from local training through synchronization. Every term is scaled to $[0,1]$ within the round, so the weights carry relative importance rather than unit conversion. The implementation holds $\kappa$ at a value calibrated once instead of running a Dinkelbach update every round, which is why the weights appear as constants. \Cref{tab:notation} collects the notation.

\begin{table}[t]
\centering
\caption{Notation.}
\label{tab:notation}
\footnotesize
\setlength{\tabcolsep}{4.5pt}
\renewcommand{\arraystretch}{1.15}
\begin{tabular}{>{\RaggedRight}p{1.6cm}|p{5.2cm}}
\toprule
\textbf{Symbol} & \textbf{Meaning} \\
\midrule
\rowcolor{crystalblue!10}
\multicolumn{2}{l}{\textbf{Clients and model}} \\
$\mathcal{K},\mathcal{M},\mathcal{E}$ & clients, sensing modalities, edge servers \\
\rowcolor{crystalblue!5}
$\mathcal{D}_k,m_k$ & local dataset and modality of client $k$ \\
$\phi_k,\boldsymbol{\theta}$ & private encoder and shared state $\{\omega,\psi\}$ \\
\rowcolor{crystalblue!5}
$d_p,d_s$ & private and shared parameter counts \\
\midrule
\rowcolor{crystalblue!10}
\multicolumn{2}{l}{\textbf{Pricing}} \\
$U_k^t$ & learning-aware utility of client $k$ at round $t$ \\
\rowcolor{crystalblue!5}
$P,G,F$ & local progress, update magnitude, fairness credit \\
$X,L$ & uplink cost and time added to the round \\
\rowcolor{crystalblue!5}
$\alpha,\gamma,\mu,\delta,\eta$ & utility weights, with $\delta=\lambda,\eta=\kappa$ \\
$\kappa,\lambda$ & value of an hour, shadow price of a byte \\
\rowcolor{crystalblue!5}
$\sigma_k^2,D^\star,\rho$ & update variance, distortion floor, loss sensitivity \\
\midrule
\rowcolor{crystalblue!10}
\multicolumn{2}{l}{\textbf{Contract and budget}} \\
$q_k^t,n_k^t$ & upload decision and local epochs \\
\rowcolor{crystalblue!5}
$b_k^t,r_k^t,z_k^t$ & uplink rate, update frequency, quantization level \\
$p_k^t,s_k^t,e_k^t$ & priority, network slice, edge server \\
\rowcolor{crystalblue!5}
$B,\Lambda_e^t$ & uplink budget and load of edge server $e$ \\
$\mathcal{S}_t$ & clients selected at round $t$ \\
\bottomrule
\end{tabular}
\end{table}

\subsection{Reading One Price Three Ways}
\label{sec:contract}

\subsubsection{The Contract}
The contract issued to a selected client is
\begin{equation}
\Gamma_k^t=\big(\underbrace{q_k^t,n_k^t}_{\text{participate}}\;\big|\;\underbrace{b_k^t,r_k^t,z_k^t}_{\text{precision}}\;\big|\;\underbrace{p_k^t,s_k^t,e_k^t}_{\text{place}}\big),
\label{eq:contract}
\end{equation}
and all eight variables follow from $U_k^t$. \emph{Participate} applies the threshold of \Cref{prop:price} and then repairs the coverage floor by admitting the best-priced client of any modality the threshold left empty, and it grants local epochs $n_k^t$ in proportion to utility, so cheap and informative clients do more local work than expensive and redundant ones.

\subsubsection{Precision by Reverse Water-Filling}
Two variables set what a selected client sends. The uplink rate is allocated in proportion to nonnegative utility,
\begin{equation}
b_k^t=\frac{[U_k^t]_+}{\sum_{j\in\mathcal{S}_t}[U_j^t]_+}\,B ,
\label{eq:rate}
\end{equation}
and the number of bits per coordinate follows from the same shadow price that produced $\delta$.

\begin{proposition}[Utility-proportional allocation]
\label{prop:pf}
For $U_k^t>0$ on $\mathcal{S}_t$, rule \Cref{eq:rate} is the unique maximizer of $\sum_{k\in\mathcal{S}_t}U_k^t\log b_k^t$ subject to $\sum_{k\in\mathcal{S}_t}b_k^t\le B$ and $b_k^t>0$, so it is the weighted proportionally fair allocation of the uplink budget.
\end{proposition}

\begin{IEEEproof}
The objective is strictly concave and increasing, so the budget binds. Setting the gradient of the Lagrangian to zero gives $U_k^t/b_k^t=\nu$ for every $k$, hence $b_k^t\propto U_k^t$, and $\sum_k b_k^t=B$ fixes $\nu$.
\end{IEEEproof}

Proportional fairness is what keeps a modality alive under a tight budget. A top-$k$ rule spends the budget on the few clients it ranks highest, whereas \Cref{eq:rate} gives every client with positive utility a positive rate.

\begin{proposition}[Reverse water-filling]
\label{prop:water}
Let the quantizer obey the high-resolution law $\mathbb{E}\|\widehat{\Delta\boldsymbol{\theta}}_k^t-\Delta\boldsymbol{\theta}_k^t\|_2^2=d_s\sigma_k^2 2^{-2z_k^t}$, where $\sigma_k^2=(G_k^t)^2/d_s$, and let reconstruction error raise the expected global loss at rate $\rho$. The precision that maximizes the bracket of \Cref{prop:price} is
\begin{equation}
z_k^{t\star}=\tfrac12\log_2\!\frac{\sigma_k^2}{D^\star},\qquad D^\star=\frac{\lambda}{16\rho\ln 2},
\label{eq:water}
\end{equation}
so every client is quantized to one common distortion floor $D^\star$ fixed by the shadow price of the uplink.
\end{proposition}

\begin{IEEEproof}
At $z$ bits per coordinate the payload is $X_k^t=d_s z/8$ bytes, so the bracket contributes $-\rho d_s\sigma_k^2 2^{-2z}-\lambda d_s z/8$, whose second derivative is negative for $\sigma_k^2>0$. Setting the first derivative to zero gives $2^{-2z}=\lambda/(16\rho\sigma_k^2\ln 2)$, which rearranges to \Cref{eq:water}.
\end{IEEEproof}

Bits are therefore not spent evenly: a client whose update is large enough to move the shared state buys precision, and one whose update is small is coarsened until its distortion reaches the same floor as everyone else's. Because $\sigma_k^2$ grows with $G_k^t$, which enters \Cref{eq:utility} positively, $z_k^{t\star}$ increases in the price a client already earned, so the number that admitted the client also sets its precision. The implementation projects $z_k^{t\star}$ onto the levels it supports and sets $r_k^t$ by the same monotone map.

\subsubsection{Place and Aggregate}
\emph{Place} derives the priority $p_k^t$ from $U_k^t$, the slice $s_k^t$ from $p_k^t$, and the edge server as the least-loaded feasible one, $e_k^t=\arg\min_{e\in\mathcal{E}}\Lambda_e^t$, which is the list-scheduling rule whose bounded makespan let the round's time be treated as separable. A selected client transmits $\widetilde{\Delta\boldsymbol{\theta}}_k^t=Q(\Delta\boldsymbol{\theta}_k^t,z_k^t)$, and the server aggregates the reconstructions by local data size,
\begin{equation}
\boldsymbol{\theta}^{t+1}=\sum_{k\in\mathcal{S}_t}\frac{|\mathcal{D}_k|}{\sum_{j\in\mathcal{S}_t}|\mathcal{D}_j|}\Big(\boldsymbol{\theta}^{t}+\widehat{\Delta\boldsymbol{\theta}}_k^t\Big).
\label{eq:aggregation}
\end{equation}

\begin{algorithm}[t]
\caption{\sys{} Round
(\colorbox{phaseI}{Profile}, \colorbox{phaseII}{Price}, and
\colorbox{phaseIII}{Assign})}
\label{alg:round}
\textbf{Input:} shared state $\boldsymbol{\theta}^{t}$, clients $\mathcal{K}$,
edge servers $\mathcal{E}$, uplink budget $B$
\begin{algorithmic}[1]
\STATE \textit{Phase I: Profile} \Comment{privacy-safe scalars}
\colorbox{phaseI}{
\parbox{0.82\columnwidth}{
\STATE broadcast $\boldsymbol{\theta}^{t}$ and run one local warm-up per client
\STATE collect $P_k^t,G_k^t,F_k^t,X_k^t,L_k^t$ and normalize them to $[0,1]$
}}

\STATE \textit{Phase II: Price} \Comment{one utility per client}
\colorbox{phaseII}{
\parbox{0.82\columnwidth}{
\STATE $U_k^t\gets\alpha P_k^t+\gamma G_k^t+\mu F_k^t-\delta X_k^t-\eta L_k^t$
by \eqref{eq:utility}
\STATE $\mathcal{S}_t\gets\{k:U_k^t>0\}$, then repair the coverage floor
of \eqref{eq:ratio}
}}

\STATE \textit{Phase III: Assign} \Comment{participate, precision, place}
\colorbox{phaseIII}{
\parbox{0.82\columnwidth}{
\FORALL{$k\in\mathcal{S}_t$}
\STATE $b_k^t\gets$ \eqref{eq:rate}; set $n_k^t,r_k^t,z_k^t,p_k^t,s_k^t$ by the
monotone maps of $U_k^t$
\STATE $e_k^t\gets\arg\min_{e\in\mathcal{E}}\Lambda_e^t$; run $n_k^t$ local
epochs; upload $Q(\Delta\boldsymbol{\theta}_k^t,z_k^t)$
\ENDFOR
\STATE aggregate the received states by \eqref{eq:aggregation}
}}
\RETURN $\boldsymbol{\theta}^{t+1}$
\end{algorithmic}
\end{algorithm}

\subsubsection{Cost of Orchestration}
Pricing is $O(|\mathcal{K}|)$ over scalars and assignment adds an $O(|\mathcal{K}|\log|\mathcal{K}|)$ sort. Quantization is $O(d_s)$ per selected client and aggregation is $O(|\mathcal{S}_t|d_s)$. Since $d_p\gg d_s$ and local optimization dominates the round, orchestration is not on the critical path: the measured communication latency of a round is 1.99 seconds against a round duration of 3.52 hours.

%% file: sections/evaluation.tex
\section{Evaluation}
\label{sec:eval}

\subsection{Experimental Setup}

\emph{Dataset and clients.}
We evaluate on the complete nuScenes \texttt{v1.0-trainval} split~\cite{Caesar2020CVPR}: 850 scenes and 34,149 synchronized samples, partitioned by scene into 28,130 training and 6,019 validation samples. Eight object categories are mapped to multi-hot scene labels, giving a multi-label recognition task over 1,094,848 annotations. Fifteen clients cover the full sensing platform: six cameras, two video streams, two partitions of the top LiDAR, and five radars. Image and video clients use ImageNet-pretrained ResNet-18 backbones, video adds a gated recurrent unit for short-term context, and LiDAR and radar clients encode bird's-eye-view representations with lightweight convolutional networks. All clients optimize an asymmetric multi-label objective with positive-class weighting using AdamW and cosine decay. The network holds two edge servers sharing a 48\,Mbps uplink budget.

\emph{Baselines.}
We compare against classical and multimodal federated learning (FedAvg~\cite{McMahan2017AISTATS}, FedProx~\cite{Li2020MLSYS}, FedCoLa~\cite{Sun2024ECCV}), heterogeneity-aware optimization (FedDyn~\cite{Acar2021ICLR}, FedADMM~\cite{Gong2022ICDE}, FedAdam~\cite{Reddi2021ICLR}), personalized and regularized methods (FedBABU~\cite{Oh2022ICLR}, FedSR~\cite{Nguyen2022NeurIPS}, MOON~\cite{Li2021CVPR}), and communication-aware methods (FedCS~\cite{Nishio2019ICC}, FedPAQ~\cite{Reisizadeh2020AISTATS}, Oort~\cite{Lai2021OSDI}). Every method uses the same split, preprocessing, loss, client organization, local budget, and validation schedule, so the reported differences come from the algorithms.

\emph{Metrics and protocol.}
The active time of a round is the training and evaluation time of its selected clients. It excludes the idle intervals between rounds and matches the recorded round wall time to within $0.01\%$. We report accuracy within a budget, meaning the highest Macro-F1 a method reaches inside a given cumulative active time, in place of the time at which it first crosses a fixed accuracy. Both quantities read the same trajectory, but the Macro-F1 curve is nearly flat near $0.66$, so the crossing time varies by 24.5 hours across seeds while accuracy at thirty hours varies by $0.6$ points. Modality coverage is the mean number of the four sensing groups present in a round. \sys{} is reported as mean and standard deviation over five seeds. The baselines are single runs of unequal length, so their peak values appear next to the round count that produced them.

\input{sections/tables/main}

\subsection{Cost and Accuracy Under a Budget}

\subsubsection{A Round Costs a Factor of Three}
\sys{} spends $3.31\pm0.13$ active hours per round against $4.85$ for the cheapest baseline and $9.78$ for the dearest, a reduction of $1.5\times$ to $3.0\times$ (\Cref{tab:main}). The spread across five seeds is $4\%$ of the mean, so cheaper rounds are a property of the contract rather than of a favorable run. Nor are they obtained by dropping expensive sensors. \sys{} keeps all four modality groups in every round of every seed. FedCS is the one baseline that also holds four, and it pays $9.50$\,MB per round to do so against $2.93$\,MB for \sys{}; the rest average $2.56$ to $3.19$ groups and leave roughly a quarter of them idle.

\subsubsection{The Lead Holds From Ten to Forty-Five Hours}
Within twenty active hours the weakest of the five seeds reaches Macro-F1 $0.6486$, above the $0.6458$ of the strongest baseline at that budget, so the lead at twenty hours does not depend on the seed. Within thirty hours the seed mean is $0.6559\pm0.0057$ against $0.6486$, and the weakest seed matches that baseline exactly. Sweeping the budget continuously, the seed mean leads every baseline from ten hours to $42.5$ hours (\Cref{fig:eval}(c)). The mechanism is visible in the round counts rather than in the optimizer: at thirty hours FedBABU has finished three rounds and \sys{} nine. Cheaper rounds convert into more optimization per hour, not into better optimization per round.

\begin{figure*}[t]
\centering
\includegraphics[width=0.99\textwidth]{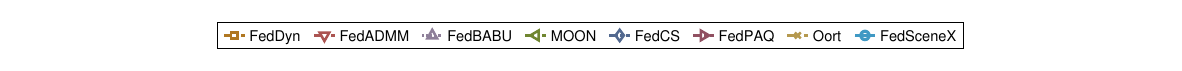}\\[-0.2em]
\begin{subfigure}[t]{0.245\textwidth}
\centering\includegraphics[width=\linewidth]{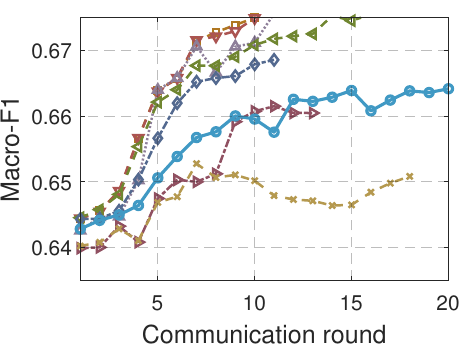}
\caption{By round}\label{fig:eval_a}
\end{subfigure}\hfill
\begin{subfigure}[t]{0.245\textwidth}
\centering\includegraphics[width=\linewidth]{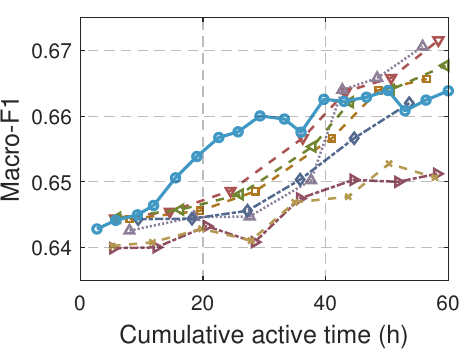}
\caption{By active time}\label{fig:eval_b}
\end{subfigure}\hfill
\begin{subfigure}[t]{0.245\textwidth}
\centering\includegraphics[width=\linewidth]{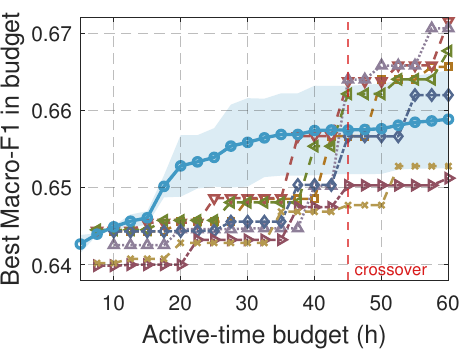}
\caption{Under a budget}\label{fig:eval_c}
\end{subfigure}\hfill
\begin{subfigure}[t]{0.245\textwidth}
\centering\includegraphics[width=\linewidth]{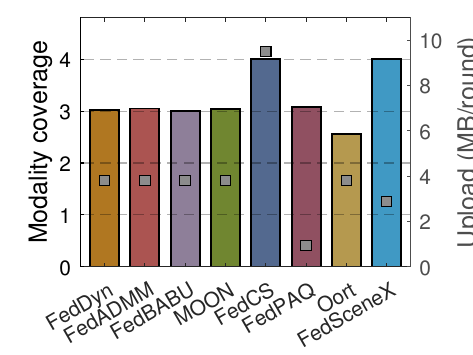}
\caption{Coverage and upload}\label{fig:eval_d}
\end{subfigure}
\caption{The same trajectories read four ways. (c) shades one standard deviation over five seeds and marks the crossover; (d) plots modality coverage as bars and upload volume as squares on the right axis.}
\label{fig:eval}
\end{figure*}

\subsection{Where the Advantage Ends}

Beyond forty-five hours the ranking reverses. At fifty hours FedBABU reaches $0.6658$ and FedDyn $0.6640$, while \sys{} is at $0.6576\pm0.0058$, and over their full trajectories the same two methods peak at $0.6798$ and $0.6819$ against $0.6713$. The crossover follows from the design. Selective transmission and reduced update precision are what make a round inexpensive, and the same coarseness limits the fine-grained refinement that late-stage optimization needs. \sys{} is therefore the right choice when a deployment has to reach a usable model under a bounded time budget, and the wrong one when final accuracy matters more than the hours spent reaching it.

\input{sections/tables/ablation}

\subsection{Which Parts of the Contract Matter}

\input{sections/tables/telemetry}

Each family of contract decisions earns its place, and each one on a different axis (\Cref{tab:ablation}). Removing the contract costs $1.54$ Macro-F1 points at thirty hours and raises the cost of a round from $3.52$ to $9.21$ hours. Replacing rate allocation, quantization, or priority with uniform policies costs the same $1.5$ points and raises the uploaded state from $2.88$ to $7.27$\,MB, because every selected client then transmits at full precision. Fairness and edge assignment are cheaper to lose, $0.8$ points each, and they act on timing rather than volume: both variants upload the same $2.91$\,MB as the full contract. Randomizing the contract while keeping the number of selected clients costs $1.12$ points and drops modality coverage from $4.00$ to $3.20$, which isolates the utility ordering rather than the client count as the active ingredient.

Removing a sensor is the one change that raises peak accuracy. Without radar the peak reaches $0.6786$ against $0.6713$, and the trace explains why the utility struggles to price it: radar produces updates nearly as large as video, with a mean norm of $0.75$ against $0.77$, yet the least learning progress of any modality, $0.0009$ against $0.0026$ for LiDAR, and the lowest client-level mAP, $0.57$ against $0.71$ for the cameras. The utility therefore pays for updates that move the shared state without improving it. Dropping radar also gives up a quarter of the sensing problem, and the resulting model still trails the full contract at thirty hours.

\subsection{The Control Channel Is Not Free}

The contract is built from per-client scalars, and those scalars describe a workload. Workloads differ systematically by modality, so the channel that makes orchestration possible also carries modality identity. We train two attackers on the server-visible telemetry alone, using rounds 1 to 60 and testing on rounds 61 to 88, a chronological split that keeps adjacent observations out of both sets. A logistic-regression attacker recovers the modality of an uploading client with accuracy $1.00$ and a random forest with $0.97$, against a majority-class baseline of $0.36$ (\Cref{tab:telemetry}).

This is not a reconstruction of sensor data. Raw observations, labels, embeddings, logits, and private encoder parameters never leave the client, and the shared state is $1.4\%$ of an image client's parameters and $24.2\%$ of a radar client's. What the telemetry exposes is a coarse client property, and it is exposed by construction: an orchestrator that prices clients by their workload must read quantities that identify the workload. Telemetry clipping, temporal aggregation, and secure aggregation of the scalars each reduce the exposure at some cost in contract quality. We report the leak rather than assume it away, since any utility-driven orchestrator inherits it.

%% file: sections/tables/main.tex
\begin{table*}[t]
\centering
\caption{Baseline comparison. \sys{}: mean$\pm$std over five seeds; other rows are single runs of the stated length. \colorbox{DeltaBg}{$\Delta$}: Macro-F1 points \sys{} gains at 30 hours.}
\label{tab:main}
\footnotesize
\setlength{\tabcolsep}{4pt}
\renewcommand{\arraystretch}{1.15}
\begin{tabular}{l|c|c|ccc|ccc|c}
\toprule
& & & \multicolumn{3}{c|}{\cellcolor{SearchBg}\textbf{Cost per round}} & \multicolumn{3}{c|}{\cellcolor{AgenticBg}\textbf{Macro-F1 within budget}} & \cellcolor{DeltaBg}\textbf{$\Delta$} \\
\cmidrule(lr){4-6}\cmidrule(lr){7-9}
\textbf{Method} & \textbf{Rnds} & \textbf{Peak Ma-F1} & \textbf{h}$\,\downarrow$ & \textbf{MB}$\,\downarrow$ & \textbf{Mod.}$\,\uparrow$ & \textbf{20\,h} & \textbf{30\,h} & \textbf{50\,h} & \\
\midrule
FedAvg~\cite{McMahan2017AISTATS}~\venue{AISTATS'17} & 100 & 0.5632 & 4.85 & 22.10 & 3.19 & 0.4068 & 0.4068 & 0.4136 & \mygreen{100}\gain{24.90} \\
\rowcolor{crystalblue!5}
FedProx~\cite{Li2020MLSYS}~\venue{MLSys'20} & 100 & 0.5605 & 4.94 & 22.10 & 3.19 & 0.3854 & 0.3854 & 0.3854 & \mygreen{100}\gain{27.05} \\
FedCoLa~\cite{Sun2024ECCV}~\venue{ECCV'24} & 52 & 0.4199 & 5.68 & -- & 3.19 & 0.4199 & \underline{0.4199} & 0.4199 & \mygreen{100}\gain{23.60} \\
\midrule
\rowcolor{crystalblue!5}
FedDyn~\cite{Acar2021ICLR}~\venue{ICLR'21} & 34 & 0.6819 & 9.78 & 3.80 & 3.03 & 0.6456 & 0.6486 & 0.6640 & \mygreen{15}\gain{0.73} \\
FedADMM~\cite{Gong2022ICDE}~\venue{ICDE'22} & 40 & 0.6807 & 9.26 & 3.80 & 3.05 & 0.6455 & \underline{0.6486} & 0.6638 & \mygreen{14}\gain{0.72} \\
\rowcolor{crystalblue!5}
FedAdam~\cite{Reddi2021ICLR}~\venue{ICLR'21} & 42 & 0.6816 & 9.19 & 3.80 & 3.07 & 0.6444 & 0.6444 & 0.6444 & \mygreen{23}\gain{1.14} \\
\midrule
FedBABU~\cite{Oh2022ICLR}~\venue{ICLR'22} & 31 & 0.6798 & 9.75 & 3.80 & 3.00 & 0.6446 & 0.6447 & 0.6658 & \mygreen{22}\gain{1.11} \\
\rowcolor{crystalblue!5}
FedSR~\cite{Nguyen2022NeurIPS}~\venue{NeurIPS'22} & 57 & 0.6796 & 5.72 & 3.80 & 3.07 & 0.6452 & 0.6452 & 0.6630 & \mygreen{21}\gain{1.06} \\
MOON~\cite{Li2021CVPR}~\venue{CVPR'21} & 65 & 0.6790 & 6.20 & 3.80 & 3.05 & 0.6458 & \underline{0.6481} & 0.6621 & \mygreen{16}\gain{0.78} \\
\midrule
\rowcolor{crystalblue!5}
FedCS~\cite{Nishio2019ICC}~\venue{ICC'19} & 11 & 0.6686 & 8.57 & 9.50 & 4.00 & 0.6444 & \underline{0.6456} & 0.6566 & \mygreen{21}\gain{1.03} \\
FedPAQ~\cite{Reisizadeh2020AISTATS}~\venue{AISTATS'20} & 13 & 0.6615 & 7.09 & 0.95 & 3.08 & 0.6400 & 0.6433 & 0.6503 & \mygreen{25}\gain{1.26} \\
\rowcolor{crystalblue!5}
Oort~\cite{Lai2021OSDI}~\venue{OSDI'21} & 18 & 0.6528 & 5.16 & 3.80 & 2.56 & 0.6429 & 0.6429 & 0.6477 & \mygreen{26}\gain{1.30} \\
\midrule
\rowcolor{crystalblue!10}
\textbf{\sys{}} & 88 & \textbf{0.6713} & \best{3.31}\std{0.13} & \textbf{2.93}\std{0.03} & \best{4.00} & \best{0.6528}\std{0.0040} & \best{0.6559}\std{0.0057} & 0.6576\std{0.0058} & -- \\
\bottomrule
\end{tabular}
\end{table*}

%% file: sections/tables/ablation.tex
\begin{table}[t]
\centering
\caption{Ablation, single runs of one seed. \colorbox{DeltaBg}{$\Delta$}: Macro-F1 points the full contract adds at 30 hours.}
\label{tab:ablation}
\footnotesize
\setlength{\tabcolsep}{4pt}
\renewcommand{\arraystretch}{1.15}
\begin{tabular}{l|ccc|c|c}
\toprule
\textbf{Variant} & \textbf{h/rnd}$\,\downarrow$ & \textbf{MB}$\,\downarrow$ & \textbf{Mod.}$\,\uparrow$ & \textbf{Ma@30\,h} & \textbf{$\Delta$} \\
\midrule
\rowcolor{crystalblue!18}
\sys{} (full) & \best{3.52} & \best{2.88} & \best{4.00} & \best{0.6600} & -- \\
\midrule
w/o contract & 9.21 & 9.50 & 4.00 & 0.6447 & \mygreen{31}\gain{1.54} \\
\rowcolor{crystalblue!5}
w/o fairness & 6.55 & 2.91 & 4.00 & 0.6518 & \mygreen{16}\gain{0.82} \\
Uniform rate & 10.46 & 7.27 & 4.00 & 0.6447 & \mygreen{31}\gain{1.54} \\
\rowcolor{crystalblue!5}
Uniform quantization & 10.30 & 7.27 & 4.00 & 0.6446 & \mygreen{31}\gain{1.54} \\
w/o priority and slice & 10.46 & 7.27 & 4.00 & 0.6446 & \mygreen{31}\gain{1.54} \\
\rowcolor{crystalblue!5}
w/o edge assignment & 6.60 & 2.91 & 4.00 & 0.6521 & \mygreen{16}\gain{0.79} \\
Random contract & 6.09 & 3.80 & 3.20 & 0.6488 & \mygreen{22}\gain{1.12} \\
\rowcolor{crystalblue!5}
w/o video & 4.17 & 2.91 & 3.00 & 0.6475 & \mygreen{25}\gain{1.25} \\
w/o LiDAR & 5.25 & 2.91 & 3.00 & 0.6503 & \mygreen{19}\gain{0.97} \\
\rowcolor{crystalblue!5}
w/o radar & 6.82 & 2.91 & 3.00 & 0.6504 & \mygreen{19}\gain{0.97} \\
\bottomrule
\end{tabular}
\end{table}

%% file: sections/tables/telemetry.tex
\begin{table}[t]
\centering
\caption{Modality inference from server-visible telemetry alone.}
\label{tab:telemetry}
\footnotesize
\setlength{\tabcolsep}{5pt}
\renewcommand{\arraystretch}{1.15}
\begin{tabular}{l|ccc}
\toprule
\textbf{Attacker} & \textbf{Accuracy} & \textbf{Macro-F1} & \textbf{Balanced acc.} \\
\midrule
Majority class & 0.3571 & 0.1316 & 0.2500 \\
\rowcolor{crystalblue!5}
Random forest & 0.9702 & 0.9706 & 0.9760 \\
Logistic regression & \best{1.0000} & \best{1.0000} & \best{1.0000} \\
\bottomrule
\end{tabular}
\end{table}

%% file: sections/conclusion.tex
\section{Conclusion}
\label{sec:conclusion}
In this paper, we have proposed \sys{}, an orchestration layer that composes each round of a same-scene multimodal federation by what the round will cost. Value-per-Hour Pricing turns the round's value-per-hour ratio into one closed-form price per client, whose negative weights are the shadow price of a byte and the worth of an hour, and the same price decides whether the client uploads, at what precision, and on which edge server. Because participation, precision, and placement ask one question, that utility replaces the separate selection, compression, and placement thresholds an edge deployment would otherwise tune by hand. On the complete nuScenes benchmark with fifteen heterogeneous clients, a \sys{} round costs $3.31$ active hours against $4.85$ to $9.78$ for twelve baselines, every seed leads the strongest baseline within a twenty-hour budget, and all four sensing modalities stay in every round. What \sys{} buys is not a better model but an earlier one, and we report the hour at which that stops being the better trade.


\balance

%% file: conference.bbl
\begin{thebibliography}{10}
\providecommand{\url}[1]{#1}
\csname url@samestyle\endcsname
\providecommand{\newblock}{\relax}
\providecommand{\bibinfo}[2]{#2}
\providecommand{\BIBentrySTDinterwordspacing}{\spaceskip=0pt\relax}
\providecommand{\BIBentryALTinterwordstretchfactor}{4}
\providecommand{\BIBentryALTinterwordspacing}{\spaceskip=\fontdimen2\font plus
\BIBentryALTinterwordstretchfactor\fontdimen3\font minus
  \fontdimen4\font\relax}
\providecommand{\BIBforeignlanguage}[2]{{%
\expandafter\ifx\csname l@#1\endcsname\relax
\typeout{** WARNING: IEEEtran.bst: No hyphenation pattern has been}%
\typeout{** loaded for the language `#1'. Using the pattern for}%
\typeout{** the default language instead.}%
\else
\language=\csname l@#1\endcsname
\fi
#2}}
\providecommand{\BIBdecl}{\relax}
\BIBdecl

\bibitem{McMahan2017AISTATS}
H.~B. McMahan, E.~Moore, D.~Ramage, S.~Hampson, and B.~Ag{\"u}era~y Arcas,
  ``{Communication-Efficient Learning of Deep Networks from Decentralized
  Data},'' in \emph{Proc. AISTATS}, 2017.

\bibitem{Kairouz2021FTML}
P.~Kairouz \emph{et~al.}, ``{Advances and Open Problems in Federated
  Learning},'' \emph{Foundations and Trends in Machine Learning}, vol.~14, no.
  1--2, pp. 1--210, 2021.

\bibitem{Duan2023COMST}
Q.~Duan, J.~Huang, S.~Hu, R.~Deng, Z.~Lu, and S.~Yu, ``{Combining Federated
  Learning and Edge Computing Toward Ubiquitous Intelligence in 6G Network:
  Challenges, Recent Advances, and Future Directions},'' \emph{IEEE
  Communications Surveys \& Tutorials}, vol.~25, no.~4, pp. 2892--2950, 2023.

\bibitem{Wu2025WASA}
B.~Wu, J.~Huang, and Q.~Duan, ``{FedTD3: An Accelerated Learning Approach for
  UAV Trajectory Planning},'' in \emph{Proc. WASA}, 2025, pp. 13--24.

\bibitem{Wu2026ARXIVForget}
B.~Wu, Z.~Ding, J.~Huang, and Y.~Zhao, ``{Forget to Improve: On-Device
  LLM-Agent Continual Learning via Budget-Curated Memory},'' arXiv preprint
  arXiv:2606.25115, 2026.

\bibitem{Wu2026ARXIVCrystalMem}
B.~Wu and J.~Huang, ``{CrystalMem: Elastic Memory for Self-Evolving LLM Agents
  via Knowledge Crystallization},'' arXiv preprint arXiv:2608.00303, 2026.

\bibitem{Feng2023ARXIV}
T.~Feng, D.~Bose, T.~Zhang, R.~Hebbar, A.~Ramakrishna, R.~Gupta, M.~Zhang,
  S.~Avestimehr, and S.~Narayanan, ``{FedMultimodal: A Benchmark for Multimodal
  Federated Learning},'' arXiv preprint arXiv:2306.09486, 2023.

\bibitem{Sun2024ECCV}
G.~Sun, Y.~Cong, J.~Dong, Q.~Wang, and J.~Liu, ``{Towards Multi-Modal
  Transformers in Federated Learning},'' in \emph{Proc. ECCV}, 2024.

\bibitem{Fang2025JSAC}
Z.~Fang, J.~Wang, Y.~Ma, Y.~Tao, Y.~Deng, X.~Chen, and Y.~Fang, ``{R-ACP:
  Real-Time Adaptive Collaborative Perception Leveraging Robust Task-Oriented
  Communications},'' \emph{IEEE Journal on Selected Areas in Communications},
  2025.

\bibitem{Fang2025ARXIV}
Z.~Fang, Y.~Guo, J.~Wang, Y.~Zhang, H.~An, Y.~Wang, and Y.~Fang, ``{Shared
  Spatial Memory Through Predictive Coding},'' arXiv preprint arXiv:2511.04235,
  2025.

\bibitem{Tchalla2026ACR}
D.~Y. Tchalla, ``{ST-Hybrid: Dynamic Graph Learning with Multi-Scale
  Spatio-Temporal Attention for Traffic Forecasting},'' \emph{ACM SIGAPP
  Applied Computing Review}, vol.~25, no.~4, pp. 35--52, 2026.

\bibitem{Wu2026ARXIVPRISM}
B.~Wu, Z.~Ding, and J.~Huang, ``{PRISM: Exposing and Resolving Spurious
  Isolation in Federated Multimodal Continual Learning},'' arXiv preprint
  arXiv:2605.01061, 2026.

\bibitem{Acar2021ICLR}
D.~A.~E. Acar, Y.~Zhao, R.~Matas, M.~Mattina, P.~Whatmough, and V.~Saligrama,
  ``{Federated Learning Based on Dynamic Regularization},'' in \emph{Proc.
  ICLR}, 2021.

\bibitem{Gong2022ICDE}
Y.~Gong, Y.~Li, and N.~M. Freris, ``{FedADMM: A Robust Federated Deep Learning
  Framework with Adaptivity to System Heterogeneity},'' in \emph{Proc. IEEE
  ICDE}, 2022, pp. 2575--2587.

\bibitem{Reddi2021ICLR}
S.~J. Reddi, Z.~Charles, M.~Zaheer, Z.~Garrett, K.~Rush, J.~Kone\v{c}n\'{y},
  S.~Kumar, and H.~B. McMahan, ``{Adaptive Federated Optimization},'' in
  \emph{Proc. ICLR}, 2021.

\bibitem{Oh2022ICLR}
J.~Oh, S.~Kim, and S.-Y. Yun, ``{FedBABU: Toward Enhanced Representation for
  Federated Image Classification},'' in \emph{Proc. ICLR}, 2022.

\bibitem{Li2021CVPR}
Q.~Li, B.~He, and D.~Song, ``{Model-Contrastive Federated Learning},'' in
  \emph{Proc. IEEE/CVF CVPR}, 2021, pp. 10\,713--10\,722.

\bibitem{Ding2025IPCCC}
Z.~Ding, J.~Huang, Q.~Duan, C.~Zhang, Y.~Zhao, and S.~Gu, ``{A Dual-Level
  Game-Theoretic Approach for Collaborative Learning in UAV-Assisted
  Heterogeneous Vehicle Networks},'' in \emph{Proc. IEEE IPCCC}, 2025, pp.
  1--8.

\bibitem{Ding2026TAAS}
Z.~Ding, J.~Huang, Y.~Zhao, and Z.~Cai, ``{Combating Knowledge Diversity and
  Catastrophic Forgetting in UAV-Assisted Collaborative Vehicular Learning: A
  Game-Theoretic Approach},'' \emph{ACM Transactions on Autonomous and Adaptive
  Systems}, 2026.

\bibitem{Ding2026ICNC}
Z.~Ding, J.~Huang, and J.~Qi, ``{Learning to Defend: A Multi-Agent
  Reinforcement Learning Framework for Stackelberg Security Game in Mobile Edge
  Computing},'' in \emph{Proc. IEEE ICNC}, 2026.

\bibitem{Wu2023MPE}
B.~Wu and W.~Wu, ``{Model-Free Cooperative Optimal Output Regulation for Linear
  Discrete-Time Multi-Agent Systems Using Reinforcement Learning},''
  \emph{Mathematical Problems in Engineering}, vol. 2023, no.~1, p. 6350647,
  2023.

\bibitem{Reisizadeh2020AISTATS}
A.~Reisizadeh, A.~Mokhtari, H.~Hassani, A.~Jadbabaie, and R.~Pedarsani,
  ``{FedPAQ: A Communication-Efficient Federated Learning Method with Periodic
  Averaging and Quantization},'' in \emph{Proc. AISTATS}, 2020.

\bibitem{Konecny2016ARXIV}
J.~Kone\v{c}n\'{y}, H.~B. McMahan, F.~X. Yu, P.~Richt\'{a}rik, A.~T. Suresh,
  and D.~Bacon, ``{Federated Learning: Strategies for Improving Communication
  Efficiency},'' arXiv preprint arXiv:1610.05492, 2016.

\bibitem{Nishio2019ICC}
T.~Nishio and R.~Yonetani, ``{Client Selection for Federated Learning with
  Heterogeneous Resources in Mobile Edge},'' in \emph{Proc. IEEE ICC}, 2019.

\bibitem{Wu2025ToN}
B.~Wu, J.~Huang, Q.~Duan, L.~Dong, and Z.~Cai, ``{Enhancing Vehicular
  Platooning With Wireless Federated Learning: A Resource-Aware Control
  Framework},'' \emph{IEEE/ACM Transactions on Networking}, 2025.

\bibitem{Huang2025TMC}
J.~Huang, B.~Wu, Q.~Duan, L.~Dong, and S.~Yu, ``{A Fast UAV Trajectory Planning
  Framework in RIS-Assisted Communication Systems With Accelerated Learning via
  Multithreading and Federating},'' \emph{IEEE Transactions on Mobile
  Computing}, pp. 1--16, 2025.

\bibitem{Wu2026COMST}
B.~Wu, J.~Huang, and S.~Yu, ``{`X of Information' Continuum: A Survey on
  AI-Driven Multi-Dimensional Metrics for Next-Generation Networked Systems},''
  \emph{IEEE Communications Surveys \& Tutorials}, vol.~28, pp. 5307--5344,
  2026.

\bibitem{Wu2023ACCESS}
B.~Wu, Z.~Cai, W.~Wu, and X.~Yin, ``{AoI-Aware Resource Management for Smart
  Health via Deep Reinforcement Learning},'' \emph{IEEE Access}, 2023.

\bibitem{Dong2026TCCN}
L.~Dong, J.~Huang, and R.~W. Heath, ``{Transformer-Based Dynamic Resource
  Allocation for Multi-Carrier NOMA Systems},'' \emph{IEEE Transactions on
  Cognitive Communications and Networking}, vol.~12, pp. 4926--4941, 2026.

\bibitem{Pudasaini2026HPSR}
U.~Pudasaini, Z.~Ding, and J.~Huang, ``{Securing Smart Agriculture with
  Communication-Efficient Federated Unlearning},'' in \emph{Proc. IEEE HPSR},
  2026, pp. 1--8.

\bibitem{Wu2025RACS}
B.~Wu, Z.~Ding, L.~Ostigaard, and J.~Huang, ``{Reinforcement Learning-Based
  Energy-Aware Coverage Path Planning for Precision Agriculture},'' in
  \emph{Proc. ACM RACS}, 2025, pp. 1--8.

\bibitem{Xing2026ACR}
C.-C. Xing, Z.~Ding, and J.~Huang, ``{A Stochastic Geometry-Based Analysis of
  SWIPT-Assisted Underlaid Device-to-Device Energy Harvesting},'' \emph{ACM
  SIGAPP Applied Computing Review}, vol.~25, no.~4, pp. 18--34, 2026.

\bibitem{Pan2023SCIS}
D.~Pan, B.-N. Wu, Y.-L. Sun, and Y.-P. Xu, ``{A Fault-Tolerant and
  Energy-Efficient Design of a Network Switch Based on a Quantum-Based
  Nano-Communication Technique},'' \emph{Sustainable Computing: Informatics and
  Systems}, vol.~37, p. 100827, 2023.

\bibitem{Fang2026ARXIVLLMSearch}
Z.~Fang, S.~F. Hu, Z.~Chang, Y.~Guo, Y.~Tao, H.~Liu, M.~Ruan, J.~Huang, and
  Y.~Fang, ``{Inference-Time Budget Control for LLM Search Agents},'' arXiv
  preprint arXiv:2605.05701, 2026.

\bibitem{Yuan2023ARXIV}
L.~Yuan, D.-J. Han, V.~P. Chellapandi, S.~H. \.{Z}ak, and C.~G. Brinton,
  ``{FedMFS: Federated Multimodal Fusion Learning with Selective Modality
  Communication},'' arXiv preprint arXiv:2310.07048, 2023.

\bibitem{Wu2026ARXIV}
B.~Wu, Z.~Ding, and J.~Huang, ``{RELIEF: Turning Missing Modalities into
  Training Acceleration for Federated Learning on Heterogeneous IoT Edge},''
  arXiv preprint arXiv:2604.04243, 2026.

\bibitem{Ding2026ARXIVEASE}
Z.~Ding, B.~Wu, and J.~Huang, ``{EASE: Federated Multimodal Unlearning via
  Entanglement-Aware Anchor Closure},'' arXiv preprint arXiv:2605.00733, 2026.

\bibitem{Lai2021OSDI}
F.~Lai, X.~Zhu, H.~V. Madhyastha, and M.~Chowdhury, ``{Oort: Efficient
  Federated Learning via Guided Participant Selection},'' in \emph{Proc. USENIX
  OSDI}, 2021.

\bibitem{Li2020MLSYS}
T.~Li, A.~K. Sahu, M.~Zaheer, M.~Sanjabi, A.~Talwalkar, and V.~Smith,
  ``{Federated Optimization in Heterogeneous Networks},'' in \emph{Proc.
  MLSys}, 2020.

\bibitem{Chen2023CVPR}
D.~Chen, J.~Hu, and V.~J. Tan, ``{Elastic Aggregation for Federated
  Optimization},'' in \emph{Proc. IEEE/CVF CVPR}, 2023, pp. 12\,187--12\,197.

\bibitem{Li2021ICLR}
X.~Li, M.~Jiang, X.~Zhang, M.~Kamp, and Q.~Dou, ``{FedBN: Federated Learning on
  Non-IID Features via Local Batch Normalization},'' in \emph{Proc. ICLR},
  2021.

\bibitem{Collins2021ICML}
L.~Collins, H.~Hassani, A.~Mokhtari, and S.~Shakkottai, ``{Exploiting Shared
  Representations for Personalized Federated Learning},'' in \emph{Proc. ICML},
  2021, pp. 2089--2099.

\bibitem{Guo2023ICML}
Y.~Guo, K.~Guo, X.~Cao, T.~Wu, and Y.~Chang, ``{Out-of-Distribution
  Generalization of Federated Learning via Implicit Invariant Relationships},''
  in \emph{Proc. ICML}, 2023, pp. 11\,905--11\,933.

\bibitem{Nguyen2022NeurIPS}
A.~T. Nguyen, P.~Torr, and S.~N. Lim, ``{FedSR: A Simple and Effective Domain
  Generalization Method for Federated Learning},'' in \emph{Proc. NeurIPS},
  2022.

\bibitem{Wu2026TNSE}
B.~Wu, Z.~Ding, and J.~Huang, ``{A Review of Continual Learning in Edge AI},''
  \emph{IEEE Transactions on Network Science and Engineering}, vol.~13, pp.
  6571--6588, 2026.

\bibitem{Wu2026ARXIVLifecycle}
B.~Wu and J.~Huang, ``{Lifecycle-Aware Federated Continual Learning in Mobile
  Autonomous Systems},'' arXiv preprint arXiv:2604.20745, 2026.

\bibitem{Wu2026ICDCS}
B.~Wu, J.~Huang, and Y.~Zhao, ``{From Alpha to Omega: Lifecycle-Aware
  Forgetting Defense in Federated Continual Learning for Planetary
  Exploration},'' in \emph{Proc. IEEE ICDCS}, 2026.

\bibitem{Caesar2020CVPR}
H.~Caesar \emph{et~al.}, ``{nuScenes: A Multimodal Dataset for Autonomous
  Driving},'' in \emph{Proc. IEEE/CVF CVPR}, 2020.

\bibitem{Wu2026MNET}
B.~Wu, J.~Huang, and Q.~Duan, ``{Real-Time Intelligent Healthcare Enabled by
  Federated Digital Twins With AoI Optimization},'' \emph{IEEE Network},
  vol.~40, no.~2, pp. 184--191, 2026.

\bibitem{Ding2026ARXIVTwinLoop}
Z.~Ding, B.~Wu, J.~Huang, and S.~Mao, ``{Application-Aware Twin-in-the-Loop
  Planning for Federated Split Learning over Wireless Edge Networks},'' arXiv
  preprint arXiv:2604.26105, 2026.

\bibitem{Ding2026ICDCS}
Z.~Ding, B.~Wu, and J.~Huang, ``{SCALE: Sensitivity-Aware Federated Unlearning
  with Information Freshness Optimization for Mobile Edge Computing},'' in
  \emph{Proc. IEEE ICDCS}, 2026.

\bibitem{Dinkelbach1967MS}
W.~Dinkelbach, ``{On Nonlinear Fractional Programming},'' \emph{Management
  Science}, vol.~13, no.~7, pp. 492--498, 1967.

\bibitem{Graham1966BSTJ}
R.~L. Graham, ``{Bounds for Certain Multiprocessing Anomalies},'' \emph{Bell
  System Technical Journal}, vol.~45, no.~9, pp. 1563--1581, 1966.

\end{thebibliography}
